\documentclass[conference]{IEEEtran}
\IEEEoverridecommandlockouts

\ifCLASSINFOpdf
 
\else
 
\fi

\usepackage{cite}
\usepackage{amsmath,amssymb,amsfonts}
\usepackage{bm}
\usepackage{algorithmic}
\usepackage{graphicx}
\usepackage{textcomp}
\usepackage{xcolor}
\usepackage{float}
\usepackage{amsthm}
\usepackage{graphicx}
\usepackage{epstopdf}
\usepackage{amsmath,bm}
\usepackage{amsfonts}
\usepackage{amssymb}
\usepackage{color}
\usepackage{multirow}
\usepackage{multicol}
\usepackage{soul,xcolor}
\usepackage{setspace}
\newtheorem{lemma}{Lemma}

\DeclareMathOperator{\tr}{tr}

\newcommand{\mathacr}[1]{\mathsf{#1}}

\newcommand{\vect}[1]{\mathbf{#1}}

\def\diag{\mathrm{diag}}

\def\tr{\mathrm{tr}}

\def\Htran{\mbox{\tiny $\mathrm{H}$}}
\def\Ttran{\mbox{\tiny $\mathrm{T}$}}

\begin{document}
\makeatletter
\newcommand*{\rom}[1]{\expandafter\@slowromancap\romannumeral #1@}
\makeatother

\title{Cell-Free Massive MIMO with Large-Scale Fading Decoding and Dynamic Cooperation Clustering
\thanks{ \"O. T. Demir and E. Bj\"ornson were supported by the FFL18-0277 grant from the Swedish Foundation for Strategic Research. L. Sanguinetti was supported by the Italian Ministry of Education and Research (MIUR) in the framework of the CrossLab project (Departments of Excellence).}
}

\author{\IEEEauthorblockN{\"Ozlem Tu\u{g}fe Demir\IEEEauthorrefmark{1},
		Emil Bj\"ornson\IEEEauthorrefmark{1}\IEEEauthorrefmark{2}, and  Luca Sanguinetti\IEEEauthorrefmark{3}}
	\IEEEauthorblockA{\IEEEauthorrefmark{1}Department of Computer Science, KTH Royal Institute of Technology, Sweden (\{ozlemtd, emilbjo\}@kth.se)}
	\IEEEauthorblockA{\IEEEauthorrefmark{2}Department of Electrical Engineering (ISY), Link\"oping University, Sweden }
		\IEEEauthorblockA{\IEEEauthorrefmark{3}Dipartimento di Ingegneria dell’Informazione, University of Pisa, Italy (luca.sanguinetti@unipi.it)}}

\maketitle

\setstretch{1.05}

\begin{abstract}
This paper considers the uplink of user-centric cell-free massive MIMO (multiple-input multiple-output) systems. We utilize the user-centric dynamic cooperation clustering (DCC) framework and derive the achievable spectral efficiency with two-layer decoding that is divided between the access points and the central processing unit (CPU). This decoding method is also known as large-scale fading decoding (LSFD). The fronthaul signaling load is analyzed and a nearly optimal second-layer decoding scheme at the CPU is proposed to reduce the fronthaul requirements compared to the optimal scheme. We also revisit the joint optimization of LSFD weights and uplink power control and show that the corresponding max-min fair optimization problem can be solved optimally via an efficient fixed-point algorithm. We provide simulations that bring new insights into the cell-free massive MIMO implementation.
\end{abstract}
\begin{IEEEkeywords}
	Cell-free massive MIMO, max-min fairness, large-scale fading decoding, dynamic cooperation clustering.
\end{IEEEkeywords}

\vspace{-2mm}
\section{Introduction}
\vspace{-2mm}

Massive MIMO (multiple-input multiple-output) is a key technology in 5G owing to its ability to substantially increase the spectral efficiency (SE) of cellular networks \cite{massivemimobook,MassiveMIMO20}. An important challenge in massive MIMO systems is the SE degradation due to the large pathloss variations and inter-cell interference, in particular for the cell-edge user equipments (UEs). An alternative network structure, known as \emph{cell-free massive MIMO}, was recently proposed as a remedy to these problems \cite{Ngo2017b, demir2021foundations}. In this type of network, all the UEs in a large coverage area are jointly served by multiple distributed access points (APs). The fronthaul connections between the APs and the central processing unit (CPU), enable the division of the processing tasks for coherently serving all the active UEs.

In the original form of cell-free massive MIMO, each UE is served by all the APs in the coverage area \cite{Ngo2017b,  Nayebi2016a, Bjornson2019c}. In \cite{Bjornson2019c}, a taxonomy is introduced in terms of uplink receiver cooperation between APs and CPU, from fully distributed (Level 2) to fully centralized (Level 4) implementations. However, in later works, scalability concerns in terms of complexity for channel estimation, data decoding/precoding, and fronthaul signaling have been raised for all these levels \cite{Interdonato2019a,Buzzi2017a, Bjornson2020} and these must be addressed to enable large-scale deployments of cell-free networks. The fact that most of the APs in a large coverage area have negligible channel gains to a particular UE has motivated researchers to propose \emph{user-centric} methods to select only a subset of APs to serve a particular UE \cite{Interdonato2019a,Buzzi2017a,Bjornson2020}. In \cite{Bjornson2020}, the dynamic cooperation clustering (DCC) framework from the network MIMO literature \cite{Bjornson2011a,
	Bjornson2013d} is incorporated into cell-free massive MIMO. The clusters of APs that serve different UEs may overlap and are selected based on the users' needs.

We note that in \cite{Bjornson2020}, both the centralized (Level 4) and fully distributed (Level 2) uplink implementations are considered with DCC. However, the Level 3 implementation (based on the taxonomy in \cite{Bjornson2019c}) has not been considered with DCC. In Level 3, the CPU applies a second layer of decoding, known as large-scale fading decoding (LSFD), to suppress interference. This distributed processing method was shown to greatly improve the SE compared to Level 2 in the original cell-free massive MIMO \cite{Nayebi2016a,Bjornson2019c}, but has not been studied in user-centric cell-free networks. Level 4 has the best SE performance among others, however it requires the computation of centralized receive combiners at the CPU, which have much higher dimensions compared to Level 3 and Level 2 local combiners and bring higher computational complexity.

This paper proposes the first user-centric Level 3 implementation for cell-free massive MIMO. The main contributions are:
\vspace{-5mm}
\begin{itemize}
	\item We provide the SE with the optimal LSFD and arbitrary local decoding at the APs for user-centric cell-free massive MIMO networks with DCC, which was not considered in previous works.
	\item  We quantify the fronthaul signaling load for the optimal LSFD vector calculation and propose a new nearly optimal LSFD scheme with reduced fronthaul requirements for user-centric cell-free massive MIMO with DCC.
	\item We revisit the max-min fair joint LSFD and uplink power control optimization, for which only a computationally demanding generic solver-based solution without any numerical results \cite{Nayebi2016a} is considered. We show that the jointly optimal solution can be obtained by a novel fixed-point algorithm with low-complexity closed-form updates. The proposed algorithm also provides the optimal max-min fair power control coefficients for Level 2 operation without LSFD where the bisection search employing numerical solvers is utilized \cite{Ngo2017b}.
\end{itemize}

\vspace{-2mm}
\section{System Model}
\vspace{-0.2mm}
We consider a cell-free massive MIMO network with $L$ APs and $K$ single-antenna UEs that are distributed over a large area. Each AP is equipped with $N$ antennas. The channel between UE $k$ and AP $l$ is represented by ${\bf h}_{kl}\in \mathbb{C}^{N}$. We use the standard block fading model \cite{massivemimobook} where ${\bf h}_{kl}$ is constant in time-frequency blocks of $\tau_c$ channel uses and takes an independent realization in each coherence block by following correlated Rayleigh fading: ${\bf h}_{kl} \sim \mathcal{N}_{\mathbb{C}}({\bf 0}_N,{\bf R}_{kl})$. The matrix ${\bf{R}}_{kl}$ specifies the spatial channel correlation between the antennas of AP $l$ and UE $k$, and models also the large-scale fading effects. Since the correlation matrices depend on the long-term statistics of channels, which are fixed during the transmission, we assume they are available wherever needed (see the practical estimation methods reviewed in  \cite{MassiveMIMO20}).

We consider the uplink, which consists of $\tau_p$ channel uses dedicated for pilots and $\tau_c-\tau_p$ channel uses for payload data. We assume the network uses so-called Level 3 operation \cite{Bjornson2019c}, which consists of two stages: local processing and LSFD. In the first stage, each AP locally estimates the channels and applies an arbitrary receive combiner to obtain local estimates of the UE data. These estimates are gathered at the CPU where they are linearly processed to perform joint detection. This is done by a variation of the LSFD strategy \cite{Nayebi2016a, Bjornson2019c,Ozdogan2019,Demir2020}, which relies only on channel statistics (since the channel estimates are not shared over the fronthaul links).

\vspace{-2mm}
\subsection{Dynamic Cooperation Clustering}
\vspace{-1mm}
We consider a user-centric cell-free network by making use of the DCC framework initially proposed in \cite{Bjornson2011a,
Bjornson2013d}. By following the notation in \cite{Bjornson2020}, let $\mathcal{M}_k\subset \{ 1, \ldots, L\}$ denote the subset of APs that serve UE $k$. These subsets are determined by the CPU for a given setup and fixed throughout the communication. We define the DCC matrices based on $\{\mathcal{M}_k;k=1,\ldots,K\}$ as 
\vspace{-1mm}
\begin{equation}\label{eq:DCC}
\vect{D}_{kl}= \begin{cases}
\vect{I}_N & \textrm{if } l \in \mathcal{M}_k \\ 
\vect{0}_{N \times N} & \textrm{if } l \not \in \mathcal{M}_k. \end{cases}
\end{equation}
Moreover, we define $\mathcal{D}_l$ as the set consisting of the UE indices that are served by AP $l$:
\vspace{-1mm}
\begin{equation}
\mathcal{D}_l = \bigg\{ k : \tr(\vect{D}_{kl})\geq 1, k \in \{ 1, \ldots, K \} \bigg\}. \vspace{-2mm}
\end{equation}

\vspace{-2mm}
\subsection{Channel Estimation}
\vspace{-1mm}
The channels from UE $k$ to AP $l$ for $l\in \mathcal{M}_k$ need to be estimated once per each coherence block. In practical large cell-free networks, the usual case is $\tau_p<K$ due to the limited pilot resources. The UEs are assigned pilots from a set of $\tau_p$ mutually orthogonal pilot sequences, and, hence, some UEs share the same pilot. Let $\bm{\varphi}_k \in \mathbb{C}^{\tau_p}$ denote the pilot sequence of UE $k$ with $||\bm{\varphi}_{k}||^2=\tau_p$ and let $\mathcal{P}_k$ denote the subset of UEs that are assigned to the same pilot as UE $k$ (including itself). Then, the received pilot signal ${\bf Z}_l\in \mathbb{C}^{N\times \tau_p}$ at AP $l$ is

\begin{equation}\label{eq:pilot}
 {\bf Z}_{l}=\sum_{k=1}^{K}\sqrt{\rho_p}{\bf h}_{kl}\bm{\varphi}_{k}^{\Ttran}+{\bf N}_l, \vspace{-2mm}
\end{equation}
where $\rho_p$ is the pilot transmit power and the elements of noise matrix ${\bf N}_{l}\in \mathbb{C}^{N \times \tau_p}$ are i.i.d. $\mathcal{N}_{\mathbb{C}}(0,\sigma^2)$ random variables. After correlating $ {\bf Z}_{l}$ with the pilot sequence $\bm{\varphi}_{k}$, a sufficient statistics for estimating the channel of UE $k$ is obtained as
\vspace{-1mm}
\begin{equation}\label{eq:suff-stats} 
 {\bf z}_{kl}=\frac{{\bf Z}_l\bm{\varphi}_k^{*}}{\sqrt{\tau_p}}=\sqrt{\tau_p\rho_p}\sum_{i \in \mathcal{P}_k}{\bf h}_{il}+{\bf n}_{kl},
 \vspace{-2mm}
\end{equation}
where ${\bf n}_{kl}={\bf N}_l\bm{\varphi}_k^{*}/\sqrt{\tau_p}$ has i.i.d. $\mathcal{N}_{\mathbb{C}}(0,\sigma^2)$ elements.
Based on $ {\bf z}_{kl}$, the minimum mean-squared error (MMSE) estimate  of ${\bf h}_{kl}$ is
 \vspace{-0.25cm}
\begin{equation}\label{eq:lmmse} 
{\bf \widehat{h}}_{kl}=\sqrt{\tau_p\rho_p}{\bf R}_{kl}\left(\tau_p\rho_p\sum_{i \in \mathcal{P}_k}{\bf R}_{il}+\sigma^2{\bf I}_N\right)^{-1}{\bf z}_{kl}.
\end{equation}

\vspace{-0.2cm}
\section{Spectral Efficiency}
\vspace{-0.1cm}

 In the uplink data phase, the received signal at AP $l$ is
 \vspace{-0.1cm}
 \begin{equation} \label{eq:received_AP}
 {\bf r}_l=\sum_{k=1}^K{\bf h}_{kl}\sqrt{\eta_k}s_k+{\bf n}_l,
 \vspace{-0.1cm}
 \end{equation}
 where $s_k\in \mathbb{C}$ is the zero-mean information signal of UE $k$ with $\mathbb{E}\left\{|s_k|^2\right\}=1$, and $\eta_k>0$ is the transmit power. The receiver noise is ${\bf n}_l \sim \mathcal{N}_{\mathbb{C}}\left({\bf 0}_N,\sigma^2{\bf I}_N\right)$. Let ${\bf v}_{kl} \in \mathbb{C}^{N}$ denote the local receive combiner for the signal of UE $k \in \mathcal{D}_l$ at AP $l$. It can be selected arbitrarily based on the the channel estimates $\{{\bf \widehat{h}}_{il}\}$ at that AP. The local estimate of $s_k$ at AP $l$ is
  \vspace{-0.2cm}
 \begin{align} \label{eq:local-data-estimate}
 \widehat s_{kl} = \vect{v}_{kl}^{\Htran} \vect{D}_{kl} {\bf r}_{l}  &= \vect{v}_{kl}^{\Htran}\vect{D}_{kl}\vect{h}_{kl}\sqrt{\eta_k} s_k \nonumber \\
 & \hspace{-0.6cm}+  \sum_{i=1,i\ne k}^{K} \vect{v}_{kl}^{\Htran} \vect{D}_{kl}\vect{h}_{il} \sqrt{\eta_i}s_i + \vect{v}_{kl}^{\Htran}\vect{D}_{kl}\vect{n}_{l}.
 \end{align}
This expression holds for any $k$, but since $\vect{D}_{kl}=\vect{0}_{N\times N}$ for $l\notin \mathcal{M}_k$ by \eqref{eq:DCC}, $ \widehat s_{kl} =0$ for APs that are not serving UE $k$. Such estimates are neither calculated nor sent to the CPU.
 
The mostly used linear receiver combiners are maximum ratio (MR), zero-forcing (ZF), regularized ZF (RZF), and MMSE \cite{massivemimobook}. Partial MMSE was recently proposed in \cite{Bjornson2020} as a ``scalable'' decoder. 
For now, we keep the choice of local combiners arbitrary. The APs in $\mathcal{M}_k$ send the data estimates to the CPU, which computes a weighted sum of $\left\{\widehat{s}_{kl}\right\}$:
 \vspace{-0.2cm}
 \begin{align} 
 \widehat{s}_k =\sum_{l\in \mathcal{M}_k}a_{kl}^*\widehat{s}_{kl}=&\sqrt{\eta_k}\left(\sum_{l\in \mathcal{M}_k} a_{kl}^* \vect{v}_{kl}^{\Htran}\vect{h}_{kl}\right)s_k \nonumber \\
 &\hspace{-1.4cm}+\sum\limits_{i=1,i\ne k}^{K}\sqrt{\eta_i}\Bigg(\sum_{l\in \mathcal{M}_k} a_{kl}^* \vect{v}_{kl}^{\Htran} \vect{h}_{il}\Bigg) s_i + n^{\prime}_k\label{eq:CPU_linearcomb-level3}
 \end{align}
 where $a_{k1},\ldots,a_{KL}$ are called the LSFD weights. In \eqref{eq:CPU_linearcomb-level3}, we used the notation $n^{\prime}_k=\sum_{l\in \mathcal{M}_k} a_{kl}^* \vect{v}_{kl}^{\Htran}\vect{n}_{l}$ for the weighted noise term. Let $l_{k1}<\ldots<l_{k|\mathcal{M}_k|}$ denote the ordered indices of the APs that serve UE $k$. We can then define $|\mathcal{M}_k|$-dimensional vectors with the receive-combined channels as
  \vspace{-0.1cm}
 \begin{equation}
 \vect{g}_{ki} = \big[  \vect{v}_{kl_{k1}}^{\Htran}\vect{h}_{il_{k1}} \ \ldots \ \vect{v}_{kl_{k|\mathcal{M}_k|}}^{\Htran}\vect{h}_{il_{k|\mathcal{M}_k|}} \big]^{\Ttran} 
 \end{equation}
and the $|\mathcal{M}_k|$-dimensional LSFD weight vector
 \begin{equation}
 \vect{a}_{k} = \big[  \ a_{kl_{k1}} \ \ldots \ a_{kl_{k|\mathcal{M}_k|}} \ \big]^{\Ttran}. 
 \end{equation}
 Then, $\widehat{s}_k$ in \eqref{eq:CPU_linearcomb-level3} can be expressed in terms of these vectors as
 \begin{equation} 
 \widehat{s}_k=\sqrt{\eta_k}\vect{a}_k^{\Htran}\vect{g}_{kk}s_k+\sum\limits_{i=1,i\ne k}^{K}\sqrt{\eta_i}\vect{a}_k^{\Htran}\vect{g}_{ki}s_i+n^{\prime}_k,  
 \end{equation}
where both the transmit powers and LSFD weights are assumed to be deterministic parameters. An achievable SE is obtained in the next lemma, which extends \cite[Prop. 2]{Bjornson2019c} by exploiting the reduced-size LSFD. 

\begin{lemma} \label{lemma:uplink-capacity}
	An achievable SE\footnote{An achievable SE is a value that is below the capacity and can be achieved in a known manner, in this case by using channel codes designed for additive white Gaussian noise channels.} of UE $k$ with LSFD and DCC is 
	\begin{equation} \label{eq:uplink-rate-expression}
	\begin{split}
	\mathacr{SE}_{k} = \frac{\tau_c-\tau_p}{\tau_c} \log_2  \left( 1 + \mathacr{SINR}_{k}  \right)
	\end{split},
	\end{equation}
	where the effective signal-to-noise-plus-ratio (SINR) is
	\begin{align} 
	{\mathacr{SINR}}_k\!\!=\!\!\frac{\eta_k\left|\vect{a}_k^{\Htran}\mathbb{E}\left\{\vect{g}_{kk}\right\}\right|^2}{\vect{a}_k^{\Htran}\!\left(\!\sum\limits_{i=1}^K\eta_{i}\mathbb{E}\left\{\vect{g}_{ki}\vect{g}_{ki}^{\Htran}\right\}\!-\!\eta_k\mathbb{E}\left\{\vect{g}_{kk}\right\}\mathbb{E}\left\{\vect{g}_{kk}^{\Htran}\right\}\!+\!\vect{F}_k\!\right)\!\vect{a}_k} \label{eq:sinr}
	\end{align}
	and the $|\mathcal{M}_k|\times|\mathcal{M}_k|$ effective noise matrix is 
	\begin{align}
\vect{F}_k=\diag\bigg(\sigma^2\mathbb{E}\Big\{ \big\Vert{\bf v}_{kl_{k1}}\big\Vert^2 \Big\}, \ldots, \sigma^2\mathbb{E}\Big\{ \big\Vert{\bf v}_{kl_{k|\mathcal{M}_k|}}\big\Vert^2 \Big\} \bigg). \nonumber
\end{align}
\end{lemma}
\begin{proof} It follows from \cite[Prop.~2]{Bjornson2019c} with the difference that the statistical vectors and matrices are defined based on DCC.
\end{proof}  
 
Notice that $\mathbb{E}\left\{\vect{g}_{kk}\right\}$ and $\mathbb{E}\left\{\vect{g}_{ki}\vect{g}_{ki}^{\Htran}\right\}$ depend only on the long-term statistics and the local combiner. With MR, closed-form expressions can be computed similar to \cite{Bjornson2020}. With any other combiner, their elements can be obtained by Monte-Carlo simulations.

We note that the denominator term $\sum_{i=1}^K\eta_{i}\mathbb{E}\left\{\vect{g}_{ki}\vect{g}_{ki}^{\Htran}\right\}-\eta_k\mathbb{E}\left\{\vect{g}_{kk}\right\}\mathbb{E}\left\{\vect{g}_{kk}^{\Htran}\right\}+\vect{F}_k$ in \eqref{eq:sinr} is positive definite. Hence, the optimal LSFD vector $\vect{a}_k$ is obtained by maximizing \eqref{eq:sinr} as a generalized Rayleigh quotient \cite[Lemma~B.10]{massivemimobook}. This yields
	\begin{align} \label{eq:LSFD-vector}
	\vect{a}_{k}^{\rm opt}\!\! =\!\! &\left( \sum\limits_{i=1}^K\eta_{i}\mathbb{E}\left\{\vect{g}_{ki}\vect{g}_{ki}^{\Htran}\right\}\!-\!\eta_k\mathbb{E}\left\{\vect{g}_{kk}\right\}\mathbb{E}\left\{\vect{g}_{kk}^{\Htran}\right\}\!+\!\vect{F}_k\!\!\right)^{\!\!-1}\!\!\!\!\!\!\mathbb{E}\!\left\{\vect{g}_{kk}\right\}\!,
	\end{align}
	which leads to
	\begin{align} \label{eq:sinr-optimum}
	\mathacr{SINR}_{k}^{\rm opt} =& \eta_{k}\mathbb{E}\left\{\vect{g}_{kk}^{\Htran}\right\}\Bigg( \sum\limits_{i=1}^K\eta_{i}\mathbb{E}\left\{\vect{g}_{ki}\vect{g}_{ki}^{\Htran}\right\}-\nonumber\\
	&\hspace{0.6cm}\eta_k\mathbb{E}\left\{\vect{g}_{kk}\right\}\mathbb{E}\left\{\vect{g}_{kk}^{\Htran}\right\}+\vect{F}_k\Bigg)^{-1}\mathbb{E}\left\{\vect{g}_{kk}\right\}.
	\end{align}
Although the LSFD concept is known in conventional cell-free massive MIMO literature \cite{Nayebi2016a, Bjornson2019c}, the optimal LSFD vector that maximizes the SE of a user-centric cell-free massive MIMO system has a reduced-size structure and only a subset of APs' statistics is involved in the computation of it, which is determined by the DCC subsets. 

The number of complex multiplications for the calculation of~\eqref{eq:LSFD-vector} can be quantified by using matrix inversion based on Cholesky factorization as in \cite[Lemma~B.2]{massivemimobook} and becomes $|\mathcal{M}_k|^2+\frac{|\mathcal{M}_k|^3-|\mathcal{M}_k|}{3}$, where $|\mathcal{M}_k|$ is the number of APs serving UE $k$. In cell-free massive MIMO without DCC, the corresponding value is $L^2+\frac{L^3-L}{3}$. In large networks, most of the APs will receive a negligible signal power from UE $k$. The SE loss of selecting $|\mathcal{M}_k| \ll L$ can thus be made negligible \cite{Bjornson2020}. Therefore, the complexity of LSFD with DCC is expected to be much lower than in the original form of cell-free massive MIMO and each UE gets its signal decoded by a subset of APs that is selected uniquely for that UE.

\subsection{Fronthaul Signaling Load}

Let us quantify the fronthaul signaling load for LSFD with DCC. AP $l$ needs to send $(\tau_c-\tau_p)|\mathcal{D}_l|$ complex scalars (i.e., the local estimates $ \widehat{s}_{kl}$ for $k \in \mathcal{D}_l$) to the CPU per coherence block. This number can be upper bounded by $(\tau_c-\tau_p) \tau_p$ if we use the pilot assignment algorithm from \cite{Bjornson2020}. 
Since this number is independent of $K$ and $L$, the fronthaul signaling per AP is manageable also in large networks,
The total number of complex scalars from all APs to the CPU per coherence block is given by $(\tau_c-\tau_p) \sum_{l=1}^L |\mathcal{D}_l|$, which is lower than $(\tau_c-\tau_p)LK$ as required by a cell-free system without DCC. 

In addition to the locally-decoded symbol estimates, the statistical parameters must also be sent to the CPU since the channels are estimated locally at the APs and not shared to CPU. We notice that the $(l,l^{\prime})$th element of $\mathbb{E}\left\{\vect{g}_{ki}\vect{g}_{ki}^{\Htran}\right\}$ is given by $\left[\mathbb{E}\left\{\vect{g}_{ki}\right\}\right]_l\left[\mathbb{E}\left\{\vect{g}_{ki}\right\}\right]_{l^{\prime}}^*$ for $ l\neq l^{\prime}$ by the independence of the channels corresponding to different APs. For a generic UE $k\in \mathcal{D}_l$, AP $l$ thus sends the following statistical parameters  to the CPU for the computation of \eqref{eq:LSFD-vector}:
\begin{itemize}
	\item $\mathbb{E}\{ {\bf v}_{kl}^{\Htran}{\bf h}_{il} \}$, for $i=1,\ldots,K$;
	\item $\mathbb{E}\{  |{\bf v}_{kl}^{\Htran}{\bf h}_{il}  |^2\}$, for $i=1,\ldots,K$;
	\item $\sigma^2\mathbb{E} \{ \| {\bf v}_{kl} \|^2 \}$
\end{itemize}
which adds up to $\left(3K+1\right)/2$ complex numbers. Each AP sends these numbers for each of its serving $|\mathcal{D}_l|$ UEs, and, hence, $\left(3K+1\right)/2\sum_{l=1}^{L} |\mathcal{D}_l|$ complex scalars are sent to the CPU in total for Level 3 operation with LSFD and DCC. Notice that the fronthaul signaling load of an original cell-free massive MIMO system is obtained by replacing $\left|\mathcal{D}_l\right|$ by $K$. This increases the fronthaul signaling requirements compared to DCC. We will quantify the gap in Section~\ref{sec:numerical-results}.

\section{A Nearly Optimal LSFD to Solve the Scalability Issue}

Although the LSFD vector in \eqref{eq:LSFD-vector} is optimal and has manageable complexity, the fronthaul signaling load for the statistical parameters increases quadratically with $KL$. This becomes a burden for large networks and makes the method unscalable. Inspired by the partial MMSE local decoding combiner from \cite{Bjornson2020}, we propose here a nearly optimal LSFD vector as a remedy to this problem. We first introduce the set 
\vspace{-0.2cm}
\begin{align}
\mathcal{S}_k^R=\left\{ i  : \sum_{l=1}^L \tr\left(\vect{D}_{kl}\vect{D}_{il}\right)\geq N \min \left(R,\left|\mathcal{M}_k\right| \right) \right\}, \label{eq:Sk}
\end{align} 
which corresponds to the indices of the UEs that are served by at least $\min \left(R,\left|\mathcal{M}_k\right| \right)$ APs that also serve UE $k$. Here, we define $R\geq1$ as a design variable and take the minimum of $R$ and $\left|\mathcal{M}_k\right|= \sum\limits_{l=1}^L\tr\left(\vect{D}_{kl}\right)\big/N$ to guarantee that the set $\mathcal{S}_k^R$ includes at least UE $k$. The rationale behind this is that, in a large cell-free setup with properly designed DCC,\footnote{The joint pilot assignment and cooperation cluster formation algorithm in \cite{Bjornson2020} can be used to select the DCC matrices to minimize pilot contamination.} the interference that affects UE $k$ is mainly due to a small subset of other UEs and these are identified by \eqref{eq:Sk}. Using this set, we can define the nearly optimal LSFD vector as
\vspace{-0.2cm}
	\begin{align} \label{eq:LSFD-vector-partial}
\vect{a}_{k}^{\rm n-opt} = &\Bigg( \sum\limits_{i\in \mathcal{S}_k^R}\eta_{i}\mathbb{E}\left\{\vect{g}_{ki}\vect{g}_{ki}^{\Htran}\right\}-\nonumber\\ 	
&\hspace{0.6cm}\eta_k\mathbb{E}\left\{\vect{g}_{kk}\right\}\mathbb{E}\left\{\vect{g}_{kk}^{\Htran}\right\}+\vect{F}_k\Bigg)^{-1}\mathbb{E}\left\{\vect{g}_{kk}\right\},
\end{align}
by including only the statistics related to the UEs with indices $i\in \mathcal{S}_k^R$. Here, the name ``nearly optimal'' refers to the fact that the new LSFD vector is constructed based on the optimal \eqref{eq:LSFD-vector} by keeping only the dominating interference terms. In this case, for a generic UE $k \in \mathcal{D}_l$ AP $l$ needs to send only $\left(3\left|\mathcal{S}_k^R\right|+1\right)/2$ complex statistical parameters. In total, the number of complex scalars to be sent to the CPU reduces to $\sum_{l=1}^L\sum_{k\in \mathcal{D}_l}\left(3\left|\mathcal{S}_k^R\right|+1\right)/2$. With tens of UEs, $\left|\mathcal{S}_k^R\right|$ will be likely much smaller than $K$. The price to pay for this fronthaul reduction is of course a performance drop in terms of SE that we will quantify by means of numerical results in Section~\ref{sec:numerical-results}. The preferable $R$ can be selected by identifying an appropriate tradeoff between complexity and performance based on the network requirements and traffic.

Note that the index set $\mathcal{S}_k^R$ in \eqref{eq:Sk} is a generalized version of the set that is considered in \cite{Bjornson2020} for only $R=1$. We additionally introduce $\min \left(R,\left|\mathcal{M}_k\right|\right)$ on the right side of \eqref{eq:Sk} to guarantee that the self-interference term in \eqref{eq:LSFD-vector-partial}  is not excluded. This is the first time that the set in \eqref{eq:LSFD-vector-partial} is used to construct a scalable LSFD vector.

\vspace{-0.1cm}
\section{Max-Min Fair Joint LSFD and Power Control}

The max-min fairness optimization criterion aims to maximize the minimum SE among all the UEs in the network and has been extensively considered in the cell-free massive MIMO literature \cite{Ngo2017b,Ngo2018a,Demir2020,Nayebi2016a,Zhang2018a}. However, the joint design of LSFD weights and uplink power control coefficients is only considered in \cite{Nayebi2016a}. The authors of \cite{Nayebi2016a} claimed that the corresponding objective function is quasi-concave and the problem can be solved with the bisection method, but how the resulting objective function involving a matrix inverse can be handled with a numerical solver is not explained. Moreover, the simulation results of \cite{Nayebi2016a} consider fixed uplink powers, thus no optimization is needed at all.
Even if there is a way to cast the problem in a manageable form, a numerical convex optimization solver is needed at each iteration of the bisection search, which will lead to a high computational complexity. Hence, there is a need to develop an efficient algorithm for the joint optimal design of LSFD weights and uplink power control. Such an algorithm is proposed first in this paper.

For the optimal LSFD vectors, we note that irrespective of the UE transmit powers $\left\{\eta_k\right\}$, the corresponding optimal SINRs are given as in \eqref{eq:sinr-optimum}. Hence, we ``only'' have to optimize the SINRs that are achieved with the optimal LSFD vectors with respect to the transmit powers $\{\eta_k\}$.
 Note that maximizing the minimum of the achievable SEs, $\{\mathacr{SE}_k\}$, is equivalent to maximizing the minimum of ${\mathacr{SINR}_k}$ by the monotonicity of the logarithm. Hence, the max-min fairness optimization problem can be written as
 \vspace{-0.2cm}
\begin{align}
& \underset{\left\{\eta_k\right\}}{\mathacr{maximize}} \ \ \underset{k}{\mathacr{min}} \  \eta_{k}\mathbb{E}\left\{\vect{g}_{kk}^{\Htran}\right\}\Bigg( \sum\limits_{i=1}^K\eta_{i}\mathbb{E}\left\{\vect{g}_{ki}\vect{g}_{ki}^{\Htran}\right\}-\nonumber\\
&\hspace{2.6cm}\eta_k\mathbb{E}\left\{\vect{g}_{kk}\right\}\mathbb{E}\left\{\vect{g}_{kk}^{\Htran}\right\}+\vect{F}_k\Bigg)^{-1}\mathbb{E}\left\{\vect{g}_{kk}\right\}  \nonumber \\
& {\mathacr{subject \ to}} \quad  0\leq \eta_k\leq \rho_u, \ \   k=1, \ldots, K, \label{eq:constraint-maxmin}
\end{align} 
where  $\rho_u$ is the maximum transmission power at each UE. We will use the following lemma to solve the above problem optimally in a few iteration steps.
\begin{lemma} \label{lemma:convergence}
	The fixed-point algorithm whose steps are outlined in Algorithm~1 converges to the optimal solution to \eqref{eq:constraint-maxmin} when the power control coefficients are initialized as arbitrary positive values.
\end{lemma}
	\begin{proof}
		The proof follows from \cite[Theorem~1]{Tan2014} by noting that the power constraints in \eqref{eq:constraint-maxmin} satisfy the monotonic constraints in \cite[Assumption~2]{Tan2014} and the concavity of the function on the right side of the equality in \eqref{eq:maxmin-step2} from \cite[Theorem 1]{Cai2011a}. Please see the discussion after Theorem 1 and the example in Equation (13) in \cite{Tan2014} for the exact details. 
	\end{proof}
\vspace{-0.3cm} 
We note that there are some seemingly similar fixed-point algorithms in the literature, used to jointly optimize beamforming vectors and power control coefficients  \cite{Rashid1998b,Dahrouj2010a,Cai2011a}. These are conceptually related, but the optimization variables are different. In this paper, we jointly optimize the long-term parameters that are kept fixed once optimized for a given setup. The novelty is to uncover that a low-complexity fixed-point algorithm converges to the joint global optimal solution of the max-min fair optimization problem in cell-free networks. Furthermore, we propose a sub-optimal fixed-point algorithm by replacing the matrix in Step 2 of Algorithm~1 by
\begin{equation} \vect{C}_k^{(j-1)}\!=\!\sum\limits_{i\in \mathcal{S}_k^R}\eta_{i}^{(j-1)}\mathbb{E}\!\left\{\vect{g}_{ki}\vect{g}_{ki}^{\Htran}\right\}\!-\!\eta_k^{(j-1)}\mathbb{E}\!\left\{\vect{g}_{kk}\right\}\!\mathbb{E}\!\left\{\vect{g}_{kk}^{\Htran}\right\}\!+\!\vect{F}_k \nonumber  
\end{equation}
for the power control with the nearly optimal LSFD. Following a similar approach as in Lemma~\ref{lemma:convergence} it can be shown that the modified fixed-point algorithm converges to the solution of the problem that is obtained by only keeping the UE indices in $\mathcal{S}_k^R$ in the objective of \eqref{eq:constraint-maxmin}:
\vspace{-0.2cm}
\begin{align}
& \underset{\left\{\eta_k\right\}}{\mathacr{maximize}} \ \ \underset{k}{\mathacr{min}} \  \eta_{k}\mathbb{E}\left\{\vect{g}_{kk}^{\Htran}\right\}\Bigg( \sum\limits_{i\in \mathcal{S}_k^R}\eta_{i}\mathbb{E}\left\{\vect{g}_{ki}\vect{g}_{ki}^{\Htran}\right\}-\nonumber\\
&\hspace{2.6cm}\eta_k\mathbb{E}\left\{\vect{g}_{kk}\right\}\mathbb{E}\left\{\vect{g}_{kk}^{\Htran}\right\}+\vect{F}_k\Bigg)^{-1}\mathbb{E}\left\{\vect{g}_{kk}\right\}  \nonumber \\
& {\mathacr{subject \ to}} \quad  0\leq \eta_k\leq \rho_u, \ \   k=1, \ldots, K. \label{eq:constraint-maxmin-2}
\end{align} 
 Although the objective function obtained in this way does not correspond to the SINR with the nearly optimal LSFD, in the special case when each UE is only affected by interference from precisely the UEs in $\mathcal{S}_k^R$, the algorithm converges to the global optimum. Simulations show that the fixed-point algorithm for both optimal and nearly optimal LSFD designs converges fast, i.e., in less than ten iterations.
\vspace{-2mm}
\begin{figure}[t!]
\begin{center}
	\linethickness{0.45mm}
	\line(1,0){250}
\end{center}
\vspace{-0.2cm}
{\bf Algorithm 1:} Fixed-Point Algorithm for Max-Min Fairness 
\vspace{-0.4cm}
\begin{center}
	\linethickness{0.15mm}
	\line(1,0){250}
\end{center}
\vspace{-0.2cm}
{\bf 1)} Initialize iteration counter: $j\leftarrow 0$. Initialize arbitrary positive power control coefficients: $\eta^{(j)}_k>0$, for $k=1,\ldots,K$.   \\
{\bf 2)} Set $j \leftarrow j+1$. Update the power control coefficients as
\begin{equation}
\eta_k^{(j)}=\frac{1}{\mathbb{E}\left\{\vect{g}_{kk}^{\Htran}\right\}\left(\vect{C}_k^{(j-1)}\right)^{-1}\mathbb{E}\left\{\vect{g}_{kk}\right\}} \label{eq:maxmin-step2}
\end{equation}
where
\begin{equation} \vect{C}_k^{(j-1)}\!=\!\sum\limits_{i=1}^K\eta_{i}^{(j-1)}\mathbb{E}\!\left\{\vect{g}_{ki}\vect{g}_{ki}^{\Htran}\right\}\!-\!\eta_k^{(j-1)}\mathbb{E}\!\left\{\vect{g}_{kk}\right\}\!\mathbb{E}\!\left\{\vect{g}_{kk}^{\Htran}\right\}\!+\!\vect{F}_k. \nonumber
\end{equation}
{\bf 3)} Scale the power control coefficients as 
\begin{equation}
\eta_k^{(j)}\leftarrow \frac{\rho_u}{\max\limits_{i} \ \eta_i^{(j)}}\eta_k^{(j)}.
\end{equation} 
{\bf 4)} Repeat Steps 2 and 3 until convergence. 
\vspace{-0.4cm}
\begin{center}
	\linethickness{0.15mm}
	\line(1,0){250}
\end{center}
\vspace{-0.8cm}
\end{figure}

\section{Numerical Results}\label{sec:numerical-results}
\vspace{-0.1cm}
In this section, we show that the user-centric Level 3 uplink operation results in a negligible performance loss compared to the original form of cell-free massive MIMO (where each UE is served by all the APs), while achieving a large reduction in complexity and fronthaul signaling. Moreover, we demonstrate the SE gain over the fully distributed Level 2 operation \cite{Bjornson2019c} and quantify the SE performance of the proposed optimal max-min fairness maximization algorithm. Note that we do not provide the performance of fully centralized (Level 4) scheme since our focus in this paper is on the improvements obtained by the proposed Level 3 scheme compared to their distributed counterparts. Level 4 certainly outperforms Level 2 and Level 3 operations by exploiting all the channel estimates at the CPU, but with a greater complexity in computing the receive combiners per coherence block \cite{Bjornson2020}.

We consider a simulation scenario with 400 random setups where $L=100$ APs with $N=4$ antennas and $K=40$ UEs are independently and uniformly dropped in a $1 \times 1$\,km square. We utilize the same propagation model (with wrap-around) and joint pilot assignment and cooperation cluster formation algorithm as in \cite{Bjornson2020}. The maximum uplink transmit power is $\rho_u=100$\,mW and we assume $\tau_c=200$ and $\tau_p=10$. For local decoding, the partial MMSE decoder in \cite{Bjornson2020} is used with either full power or max-min fair (MMF) power control. The power control coefficients in Algorithm~1 are all initialized as $\rho_u$ and the termination is done when $\frac{\left\vert\max_{k} \mathrm{SINR}_k-\min_{k} \mathrm{SINR}_k\right\vert}{\left\vert\max_{k} \mathrm{SINR}_k\right\vert}\le 0.001$. For the nearly-optimal LSFD case, $\mathrm{SINR}_k$ is replaced by the objective function of \eqref{eq:constraint-maxmin-2}.

\begin{figure}[t!]
	\vspace{0.1cm}
	\begin{center}
		\includegraphics[trim={2.2cm 0cm 3.2cm 0.6cm},clip,width=8.8cm]{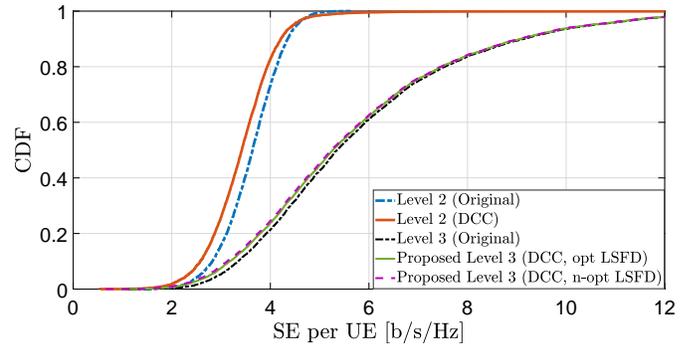}
		\vspace{-6mm}
		\caption{The CDF of the SE per user for full uplink power, i.e., $\eta_k=\rho_u$, $\forall k$.} \label{fig:sim1}
	\end{center}
\vspace{-6mm}
\end{figure}

Fig.~\ref{fig:sim1} shows the cumulative distribution function (CDF) of the SE per UE with Level 3 and Level 2 implementations. Full transmit power is used, i.e., $\eta_k=\rho_u, \forall k$. Note that very high SE can be obtained for some UEs but in practice the SE might be limited by the implementation constraints such as the largest modulation format and coding rate. We first notice that there is a very small gap between the original form of cell-free massive MIMO \cite{Bjornson2019c} and the proposed user-centric one with DCC. The advantage of proposed methods is quantified in Table~\ref{tab:signaling}, where the required number of complex multiplications for LSFD vector computation, fronthaul signaling in each coherence block and statistical parameters sent from all the APs to the CPU are compared. There are approximately 70-fold and five-fold reductions in the LSFD computation complexity and fronthaul signaling load, respectively, with DCC. Yet, the gap in SE is negligible, particularly at Level 3. 

Secondly, the proposed Level 3 operation with either optimal LSFD (denoted by ``opt'') or nearly optimal LSFD (denoted by ``n-opt'') provides significantly higher SE than Level 2 \cite{Bjornson2020} does. For n-opt LSFD, the parameter $R$ in constructing the sets in \eqref{eq:Sk} is set to 5. There is an approximately 50\% SE gain at the median point, where the CDF is 0.5. Interestingly, the proposed DCC with nearly optimal LSFD with $R=5$ performs almost the same SE as the optimal one with an additional two-fold reduction in statistical parameter sharing through the fronthaul as can be seen from Table~\ref{tab:signaling}. It is worth mentioning that the actual fronthaul load will also depend on the number of bits used in quantization. Nevertheless, the number of complex scalars in Table~\ref{tab:signaling} still represent significantly  reduced load thanks to the proposed method.

\begin{figure}[t!]
		\vspace{0.1cm}
	\begin{center}
		\includegraphics[trim={2.2cm 0cm 3.2cm 0.6cm},clip,width=8.8cm]{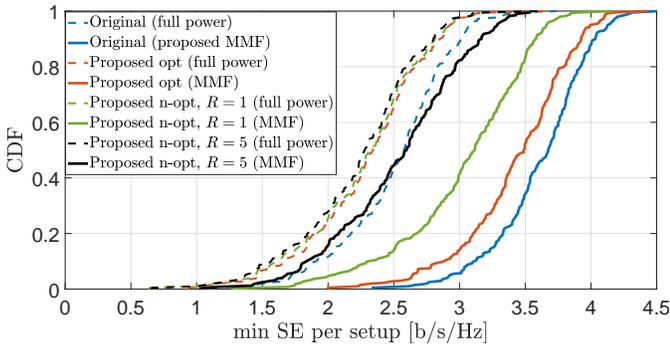}
		\vspace{-6mm}
		\caption{The CDF of the minimum SE per setup for full uplink power and MMF power control.} \label{fig:sim2}
	\end{center}
\vspace{-7mm}
\end{figure}
\begin{table}[t]
	\footnotesize
	\caption{For a generic UE, average number of complex multiplications to compute the LSFD vectors (Column 1), average number of complex scalars that are exchanged with the CPU for uplink signals in each coherence block (Column 2) and for the statistical parameters in each setup (Column 3).}  \label{tab:signaling}
	\vspace{-0.2cm}
	\centering
	\begin{tabular}{|c|c|c|c|} 
		\hline
		{} & {\shortstack{ LSFD  \\Computation}}&{\shortstack{ Uplink \\Signals}}&{Statistics}   \\      \hline\hline 
		
		Original, opt&   343300  & 19000  & 6050    \\ \hline    
		Proposed opt &   5025  & 4085 & 1301   \\   \hline   
	    Proposed n-opt ($R=1$)&   5025  & 4085 & 1144  \\   \hline   
	     Proposed n-opt ($R=5$)&   5025  & 4085 & 640  \\   \hline  
	\end{tabular}
\vspace{-4mm}
\end{table}

To see the fairness improvement by using the proposed MMF power control in Level 3 operation, in comparison to the full uplink power (i.e., $\eta_k=\rho_u$, $\forall k$), we plot the CDF of the minimum SE among all the UEs per setup in Fig.~\ref{fig:sim2}. For the original and the proposed DCC-based cell-free massive MIMO, we use Algorithm~1. However, in the DCC case with nearly optimal LSFD, the modified sub-optimal fixed-point algorithm is used. In both cases, there is a significant improvement in the worst SE in the network with the proposed MMF power control. We further note that for MMF, the reduction in the SE for n-opt LSFD with $R=5$ is higher than the case with $R=1$. However, from Table 1 we see that the exchange load in terms of statistical parameters is much less with $R=5$.

\vspace{-0.1cm}
\section{Conclusions}
\vspace{-0.2cm}
In this paper, we take a new look at cell-free massive MIMO with Level 3 uplink operation from a user-centric DCC perspective. We generalize the SE expressions and compute the fronthaul signaling load by incorporating DCC. We propose a new nearly optimal LSFD scheme with less fronthaul requirements and show that it performs almost the same as the optimal LSFD. We also take a look at the joint design of LSFD and MMF (max-min fair) power control for which no algorithm is provided in the literature. We propose the first fixed-point algorithm that obtains the global optimum with closed-form updates. We show numerically that the user-centric cell-free massive MIMO with DCC attains almost the same SE as the original form, but with much lower complexity and signaling requirements. Furthermore, the proposed optimal and sub-optimal MMF design improves the worst SE substantially in comparison to the full uplink power case for each UE.

\bibliographystyle{IEEEtran}

\bibliography{IEEEabrv,refs}

\end{document}